\documentclass[journal,twocolumn]{IEEEtran}
\ifCLASSINFOpdf
   \usepackage[pdftex]{graphicx}
   \graphicspath{{../figures/}{../jpeg/}}
   \DeclareGraphicsExtensions{.pdf,.jpeg,.png}
\else
   \usepackage[dvips]{graphicx}
   \DeclareGraphicsExtensions{.eps}
\fi
\usepackage[cmex10]{amsmath}
\usepackage{amssymb}
\usepackage{algorithmic}
\usepackage[tight,footnotesize]{subfigure}
\usepackage[font=footnotesize]{subfig}
\usepackage{algorithm}
\usepackage{bm}
\usepackage{psfrag}
\usepackage{tikz}
\usepackage{soul}
\usepackage{xcolor}
\usepackage{cite}
\begin{document}
%
\title{A Modified Levenberg-Marquardt Method for the Bidirectional Relay Channel}
\author{Guido~Dartmann,
	Ehsan~Zandi,
        Gerd~Ascheid,~\IEEEmembership{Senior Member,~IEEE}
\thanks{\footnotesize 
Guido Dartmann and Gerd Ascheid are with the Institute for Communication Technologies and Embedded Systems, RWTH Aachen University, Templergraben 55, 52056 Aachen, Germany; 
tel.: +49 241 80 27871; fax: +49 241 80 22195; email: \{guido.dartmann, gerd.ascheid\}@ice.rwth-aachen.de. 

Ehsan Zandi was with the Institute for Communication Technologies and Embedded Systems and is now with the Institute for Theoretical 
Information Technology, RWTH Aachen University, Templergraben 55, 52056 Aachen, Germany; e-mail: ehsan.zandi@ti.rwth-aachen.de. }}

\maketitle
\newtheorem{proposition}{Proposition}
\begin{abstract}
This paper presents an optimization approach for a system consisting of multiple bidirectional links over a two-way amplify-and-forward relay. It is desired to improve the fairness of the system.
All user pairs exchange information over one relay station with multiple antennas.
Due to the joint transmission to all users, the users are subject to mutual interference. A mitigation of the interference can be achieved by max-min fair precoding optimization where the relay is subject to a sum power constraint. 
The resulting optimization problem is non-convex.  
This paper proposes a novel iterative and low complexity approach based on a modified Levenberg-Marquardt method to find near optimal solutions. 
The presented method finds solutions close to the standard convex-solver based relaxation approach.
\end{abstract}


\section{Introduction}

\IEEEPARstart{T}{he} bidirectional relay channel is a well-known cooperative wireless communication scenario where $M$ pairs of users exchange information over an amplify-and-forward relay. 
The relay cannot jointly receive and transmit, hence, it can be seen as a half-duplex relay. In classical systems, the users compete with each other for the wireless resources. 
A cooperative system can increase the fairness and/or system throughput with a centralized coordination at the expense of required global channel knowledge of all cooperative links.
The entire transmission from the sources to the destinations via a relay consists of two phases. 
In the first phase the users transmit to the relay station. Then, the relay combines the signals to a new signal. In the second phase the relay forwards the combined and amplified signal to the users. 

\subsection{Related Work:}
The first works regarding cooperative communication via the relay channels consider so-called one-way relay channels where the transmission is possible only in one direction.
The work of~\cite{HavaryTSPROC08} presents an optimal solution for a transmission of one source node to a destination node over multiple one-way relays each equipped with a single antenna. 
This one-way half-duplex relay system has the disadvantage of a capacity loss due to the half-duplex transmission at the relay nodes: In the first phase the source node transmits the signal to the relay, 
then the relay forwards the signal to the destination. The uplink transmission needs further two phases.
The two-way relay channel can overcome this capacity loss. Such a system combines the uplink and downlink transmission in two hops. 
Several works \cite{VazeITW2009,HavaryTSPROC2010,JingTSPROC2012,WangTCOM2012,ShahbazPanahiTSPROC2012} investigated the cooperative communication over a bidirectional relay channel with two users. 
In this single link scenario, an optimal solution can be obtained \cite{HavaryTSPROC2010,ShahbazPanahiTSPROC2012}. 
The generalization of the single link scenario is the multiuser bidirectional relay channel where multiple users compete for the wireless resources 
\cite{ChenTWCCOM2009,BournakaGC2011,SchadCAMSAP2011,WangTWCCOM2012,ZhangWSA2012,TaoTSPROC2012}.  
The transmission can be achieved over multiple relays each equipped with a single antenna as in \cite{SchadCAMSAP2011,BournakaGC2011}, or over a single relay equipped with multiple antennas 
as in \cite{ChenTWCCOM2009,WangTWCCOM2012,ZhangWSA2012,TaoTSPROC2012}. 
In this multi-link scenario it is often desired to improve the fairness among the users by optimizing the precoding vectors \cite{SchadCAMSAP2011,BournakaGC2011,TaoTSPROC2012}. 
The resulting problem is called max-min signal-to-interference-plus-noise ratio (SINR) optimization and is non-convex, in general. Several algorithms are based on convex relaxations with convex solvers \cite{TaoTSPROC2012}.
Two-way relaying is also termed as analog network coding. The work \cite{Song2010} investigates a scenario with a single source and a destination and selects best relay from a set of 
multiple relays based on the minimum symbol error rate.
Also the work \cite{Zhang2009b} considers the scenario with a single source and destination. However, the authors in \cite{Zhang2009b} consider beamforming at a single RS with multiple antennas. 
In \cite{Zhang2009a}, the authors extend their work to a scenario with multiple users and BSs. In contrast to our paper, the authors in \cite{Zhang2009a} investigate the power minimization problem. 
Their approach is mainly based on convex solvers.

\subsection{Contribution:}
The power control problem at the users for fixed relay precoders corresponds to a unicast power control problem which can be solved efficiently \cite{TanTSPROC11}.
Therefore, we do not focus on the user power control problem in this paper.
On the other hand, the max-min SINR relay precoder optimization problem is non-convex. 
However, it can be straightforwardly relaxed to a quasi-convex problem and solved via a bisection over convex feasibility check problems. 
These convex solvers often have a bad worst-case complexity \cite{KARIPIDIS08}. 
Therefore, this paper proposes an iterative algorithm, without the requirement of a convex solvers, e.g., \cite{Zhang2009a}, based on the Levenberg-Marquardt (LM) method with line search. 
To the best of our knowledge, there exists no SINR balancing approach in the literature which is based on the Levenberg-Marquardt method.
The derived approach requires an estimation of the balanced SINR, therefore, this paper also presents a novel closed form solution for the upper bound of the balanced SINR. 
The convergence of the presented LM method is proved and numerical results show only a small performance loss compared to the convex solver based methods.

\section{Data Model and System Setup}
The most important notations of this paper are summarized in Table \ref{notation}.
 \begin{table} [tch]
\begin{center}
 \caption [Summary of all notations in the paper]{Summary of all notations in the paper.} \label{notation}
  \begin{tabular}{| c || c | }
   \hline
Symbol/Notation & Meaning  \\ \hline \hline
$\mathbb{R}$ & set of all real numbers \\  \hline
$\mathbb{R}^+$ & set of all non-negative real numbers \\  \hline
$\mathbb{C}$ & set of all complex numbers \\  \hline
$\mathbb{R}^{m\,\times\,n}$ & set of all real-valued matrices of size $m\times n$ \\  \hline 
$\mathbb{C}^{m\,\times\,n}$ & set of all complex-valued matrices of size $m\times n$ \\  \hline
$|\;.\;|$ & absolute value/magnitude \\ \hline
$||\;.\;||$ & Euclidean norm, Frobenius norm  \\ \hline
$\operatorname{Tr}\{\;.\;\}$ &  trace of a matrix    \\ \hline
$[\mathbf{A}]_{i,j}$ & element $i,j$ of matrix $\mathbf{A}$ \\  \hline
$[\mathbf{A}]_{i,:}$ & $i^{th}$ row of matrix $\mathbf{A}$ \\  \hline
$[\mathbf{A}]_{:,j}$ & $j^{th}$ column of matrix $\mathbf{A}$\\  \hline
$\mathbf{I}_m$ &  identity matrix of size $m\times m$ \\  \hline
$(\;.\;)^H$ &  Hermitian operator  \\  \hline
$(\;.\;)^T$ &  transpose operator \\ \hline
$\mathbb{E}\{\;.\;\}$ & expected value  \\  \hline
$\otimes$ &  Kronecker product  \\  \hline
$\mathcal{O}$ & big $O$ notation \\  \hline
$\succeq$ &  positive semi-definite\\  \hline
$\lambda_{max}(\;.\;)$  &  maximum eigenvalue of a matrix \\  \hline
$\lambda_{min}(\;.\;)$  &  minimum eigenvalue of a matrix\\  \hline
$\operatorname{vec}(\mathbf{A})$ &  vectorized version of a matrix  \\  \hline
$\mathcal{R}\{\mathbf{A}\}$ &  real part matrix  \\  \hline
$\mathcal{I}\{\mathbf{A}\}$ &  imaginary part matrix  \\  \hline
  \end{tabular}
\end{center}

\end{table}

  \begin{figure}[t]
  \begin{center}
      \includegraphics[scale=0.55]{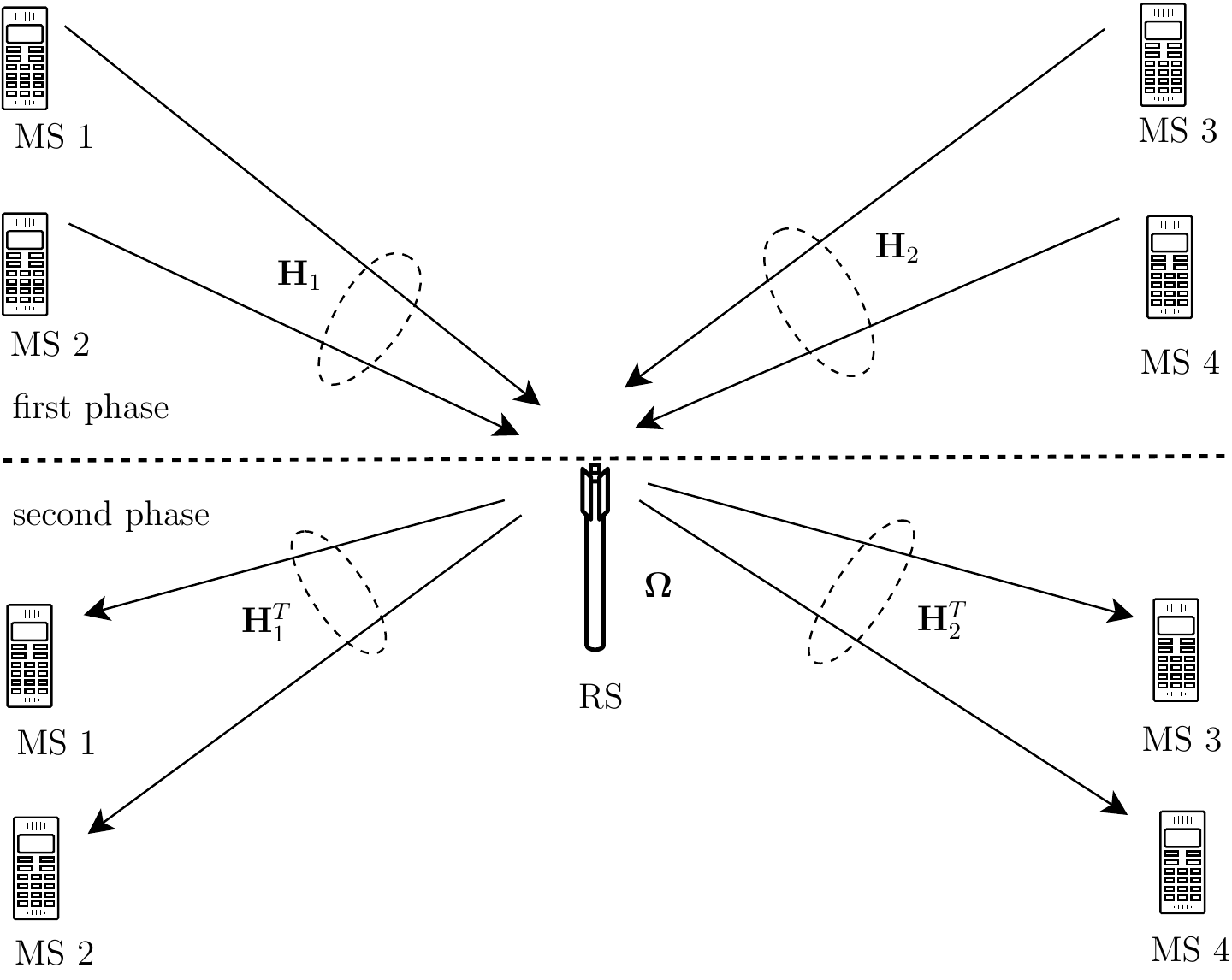}
      \caption{\small System setup of the considered network with a two-way RS.}
      \label{system_fig}
  \end{center}
 \end{figure}
\begin{figure*}[!t]

\normalsize
\begin{equation}\label{SINR1}
\gamma_i^t(\boldsymbol{\Omega})=\frac{|[\mathbf{A}_{t}\,\boldsymbol{\Omega}\,\mathbf{B}_{t}]_{i,i}|^2}{\sum_{\substack{j=1\\j\neq i}}^M |[\mathbf{A}_{t}\,\boldsymbol{\Omega}\,\mathbf{B}_{t}]_{i,j}|^2+
\sum_{\substack{j=1\\j\neq i}}^M |[\mathbf{A}_{t}\,\boldsymbol{\Omega}\,\mathbf{C}_{t}]_{i,j}|^2+\mathbb{E}\{||[\mathbf{A}_{t}\, \boldsymbol{\Omega}\mathbf{n}_R]_{i,:}||^2\}+\mathbb{E}\{|[\mathbf{n}_{t}]_i|^2\}}.
\end{equation}
\begin{equation}\label{SINR2}
\gamma_i^t(\boldsymbol{\Omega})=\frac{|[\mathbf{A}_{t}\,\boldsymbol{\Omega}\,\mathbf{B}_{t}]_{i,i}|^2}{\sum_{\substack{j=1\\j\neq i}}^M |[\mathbf{A}_{t}\,\boldsymbol{\Omega}\,\mathbf{B}_{t}]_{i,j}|^2+
\sum_{\substack{j=1\\j\neq i}}^M |[\mathbf{A}_{t}\,\boldsymbol{\Omega}\,\mathbf{C}_{t}]_{i,j}|^2+\sigma_R^2||[\mathbf{A}_{t}\, \boldsymbol{\Omega}]_{i,:}||^2+\sigma^2}.
\end{equation}
\hrulefill
\end{figure*}

This paper considers a system consisting of two sets of users $\mathcal{U}_1$ and  $\mathcal{U}_2$. Each set contains $M$ users where each user is equipped with a single antenna. 
The relay station (RS) is equipped with $N_R$ antennas. Each user of one set exchanges information with one user from the other set. 
Figure \ref{system_fig} depicts the setting of the considered system including all notations of channels and precoding matrices.  
The users of first set transmit the signal vector $\mathbf{x}_{1}\in \mathbb{C}^{M \times 1}$ and the users of the second set transmit the corresponding signal vector $\mathbf{x}_{2}\in \mathbb{C}^{M \times 1}$.
In the first phase, all $2M$ users transmit to the RS. The received signal at the RS is given by
\begin{equation}
 \mathbf{r}_R=\mathbf{H}_1   \mathbf{x}_{1}+\mathbf{H}_2\mathbf{x}_{2}+\mathbf{n}_R.
\end{equation}
Let $t,\bar{t} \in \{1,2\}$ and $\bar{t}\neq t$, in second phase, the relay station transmits the signal $\mathbf{s}_R=\boldsymbol{\Omega} \mathbf{r}_R$.
The users of the set with index $t$ receive the signal
\begin{align}
 \mathbf{r}_{t}&= \mathbf{H}_t^T \mathbf{s}_R+\mathbf{n}_{t}=    \mathbf{H}_t^T \boldsymbol{\Omega} [\mathbf{H}_t   \mathbf{x}_{t}+\mathbf{H}_{\bar{t}}\mathbf{x}_{\bar{t}}+\mathbf{n}_R]+\mathbf{n}_{t}.\nonumber \\
&= \mathbf{H}_t^T \boldsymbol{\Omega} \mathbf{H}_t   \mathbf{x}_{t}+   \mathbf{H}_t^T \boldsymbol{\Omega}\mathbf{H}_{\bar{t}}\mathbf{x}_{\bar{t}}+ \mathbf{H}_t^T \boldsymbol{\Omega}\mathbf{n}_R+ \mathbf{n}_{t}.
\end{align}
With the definitions $\mathbf{A}_{t}=  \mathbf{H}_t^T$, $\mathbf{B}_{t}=\mathbf{H}_{\bar{t}}$, and $\mathbf{C}_{t}=\mathbf{H}_t$, the received signal can be simplified to:
\begin{align}
 \mathbf{r}_{t} = \mathbf{A}_{t}\,\boldsymbol{\Omega}\,\mathbf{C}_{t}\,\mathbf{x}_{t}+ \mathbf{A}_{t}\,\boldsymbol{\Omega}\,\mathbf{B}_{t}\,\mathbf{x}_{\bar{t}}+\mathbf{A}_{t}\, \boldsymbol{\Omega}\,\mathbf{n}_R+\mathbf{n}_{t}.
\end{align}
A useful performance measure is the SINR given in Eq. \eqref{SINR1}.  
Notice the numerator of this fraction corresponds to the useful signal for the user of desire. The back-propagated self-interface, $[\mathbf{A}_{t}\,\boldsymbol{\Omega}\,\mathbf{C}_{t}]_{i,i}$, can be canceled out
assuming that complete channel information is available at each node \cite{TaoTSPROC2012}.
Assuming the noise vectors $\mathbf{n}_{t}$ and $\mathbf{n}_R$ are Gaussian, independent and identically distributed (iid) with zero mean and have the variance 
$\mathbb{E}\{\mathbf{n}_{t}\mathbf{n}_{t}^H\}=\sigma^2\mathbf{I}$ and $\mathbb{E}\{\mathbf{n}_R\mathbf{n}_R^H\}=\sigma_R^2\mathbf{I}$, 
the weighted noise term can be simplified to 
\begin{align}\label{line_noise}
 \mathbb{E}\{||[\mathbf{A}_t\, \boldsymbol{\Omega}\,\mathbf{n}_R]_i||^2\}&= \mathbb{E}\{||[\mathbf{A}_t\, \boldsymbol{\Omega}]_{i,:}\,\mathbf{n}_R||^2\} \\ \nonumber
&= \sigma_R^2\,[\mathbf{A}_t\, \boldsymbol{\Omega}]_{i,:}([\mathbf{A}_t\, \boldsymbol{\Omega}]_{i,:})^H \\ \nonumber
&= \sigma_R^2||[\mathbf{A}_t\, \boldsymbol{\Omega}]_{i,:}||^2. 
\end{align}
Furthermore, we can simplify $\mathbb{E}\{|[\mathbf{n}_{t}]_i|^2\}=\mathbb{E}\{[\mathbf{n}_{t}]_i ([\mathbf{n}_{t}]_i)^H\}=\sigma^2$, $\forall \; t\in \{1,2\}$  $\forall \; i\in \{1,\ldots, M\}$
Hence, we can rewrite \eqref{SINR1} to \eqref{SINR2}
\section{Optimization of the Relay Transmitter}

\subsection{Optimization Problem}
It is desired to improve the fairness among users. This approach can be expressed by the following optimization problem.
\begin{align} 
 \gamma^*=\max_{\boldsymbol{\Omega}}\;\;  \min_{\substack{i \in \{1,\ldots M\} \\ t\in \{1,2\}}} & \;\gamma^t_i(\boldsymbol{\Omega}) \label{MBPDL1_twoway}\\
\text{s.t.} &  \hspace{0.4cm} \operatorname{Tr}\{\boldsymbol{\Omega}\,\mathbf{Y}\,\boldsymbol{\Omega}^H\} \leq P \nonumber
\end{align}
where $\operatorname{Tr}\,\{\boldsymbol{\Omega}\mathbf{Y}\,\boldsymbol{\Omega}^H\}$ is maximum allowed transmit power at relay station
and $\mathbf{Y}\,=\mathbf{H}_1\,\mathbf{H}_1^H+\mathbf{H}_2\,\mathbf{H}_2^H+\sigma_R^2\,\mathbf{I}$. 
Problem \eqref{MBPDL1_twoway} is non-convex, due to the non-convex objective function. In what follows, we show that problem \eqref{MBPDL1_twoway} is a fractional program with quadratic numerators and denominators.
Similar to \cite{TaoTSPROC2012}, with $\boldsymbol{\omega}=\operatorname{vec}(\boldsymbol{\Omega})$ and $\mathbf{N}^t_i= \sigma_R^2\,\operatorname{diag}(\underbrace{[\mathbf{A}_t]_{i,:}^H[\mathbf{A}_t]_{i,:},\ldots,[\mathbf{A}_t]_{i,:}^H[\mathbf{A}_t]_{i,:}}_{\text{$N_R$ times}})$,
the noise term can be written as:
\begin{align}\label{noisePower}
  \sigma_R^2\,[\mathbf{A}_t\, \boldsymbol{\Omega}]_{i,:}([\mathbf{A}_t\, \boldsymbol{\Omega}]_{i,:})^H &=\sigma_R^2\,\sum_{k=1}^{N_R}[\boldsymbol{\Omega}]_{:,k}^H[\mathbf{A}_t]_{i,:}^H[\mathbf{A}_t]_{i,:}[\boldsymbol{\Omega}]_{:,k}\\ 
&=\boldsymbol{\omega}^H\mathbf{N}^t_i\boldsymbol{\omega}.\nonumber
\end{align}
The signal terms can be simplified as well. With $\mathbf{q}^t_{i,j}=[[\mathbf{A}_t]_{i,:}[\mathbf{B}_t]_{1,j},\ldots,[\mathbf{A}_t]_{i,:}[\mathbf{B}_t]_{N_R,j}]^H$ 
and $\mathbf{Q}^t_{i,j}=\mathbf{q}^t_{i,j}\mathbf{q}^{t^H}_{i,j}$, the interference is
\begin{align} \label{interferencePower}
[\mathbf{A}_t\,\boldsymbol{\Omega}\,\mathbf{B}_t]_{i,j} &=\sum_{k=1}^{N_R}[\mathbf{A}_t]_{i,:}[\mathbf{B}_t]_{k,j}[\boldsymbol{\Omega}]_{:,k}\\ \nonumber
\Rightarrow\; |[\mathbf{A}_t\,\boldsymbol{\Omega}\,\mathbf{B}_t]_{i,j}|^2&=\boldsymbol{\omega}^H\mathbf{Q}^t_{i,j}\boldsymbol{\omega}~~;~~\forall i,j\in \{1,\ldots,M\}.
\end{align}
Similarly, we can write 
\begin{align}\label{SameSideIUI}
|[\mathbf{A}_t\,\boldsymbol{\Omega}\,\mathbf{C}_t]_{i,j}|^2 &=\boldsymbol{\omega}^H\mathbf{S}^t_{i,j}\boldsymbol{\omega}~~;~~\forall i,j\in \{1,\ldots,M\}
\end{align}
where \mbox{$\mathbf{s}_{i,j}=[[\mathbf{A}_t]_{i,:}[\mathbf{C}_t]_{1,j},\ldots,[\mathbf{A}_t]_{i,:}[\mathbf{C}_t]_{N_R,j}]^H$} and  $\mathbf{S}^t_{i,j}=\mathbf{s}_{i,j}\mathbf{s}_{i,j}^H$.
The terms \eqref{noisePower}, \eqref{interferencePower} and \eqref{SameSideIUI} can be combined to
\begin{equation}\label{matrixP}
 \boldsymbol{\omega}^H\mathbf{P}^t_i\boldsymbol{\omega}=\boldsymbol{\omega}^H\Big(\mathbf{N}^t_i+\sum_{\substack{j=1\\j\neq i}}^{M}(\mathbf{Q}^t_{i,j}+\mathbf{S}^t_{i,j})\big)\boldsymbol{\omega}^H.
\end{equation}
It can also be shown that $\operatorname{Tr}\,\{\boldsymbol{\Omega}\,\mathbf{Y}\,\boldsymbol{\Omega}^H\}\,=\boldsymbol{\omega}^H\,(\mathbf{Y}^T\,\otimes\,\mathbf{I}_{N_R})\,\boldsymbol{\omega}$, in which $\otimes$ 
notifies the Kronecker product of two matrices. Therefore, using \eqref{interferencePower} and \eqref{matrixP}, the original max-min optimization \eqref{MBPDL1_twoway} 
problem is also given as the following fractional program
\begin{align} \label{MBPDL1_twoway_simple}
 \gamma^*= & \max_{\boldsymbol{\omega}}\;\;  \min_{\substack{i \in \{1,\ldots,M\}\\ t\in\{1,2\}}}  
\frac{\boldsymbol{\omega}^H\mathbf{Q}^t_{i,i}\boldsymbol{\omega}}{\boldsymbol{\omega}^H\mathbf{P}^t_i\boldsymbol{\omega}+\sigma^2} \\
\text{s.t.} &  \hspace{0.4cm} \boldsymbol{\omega}^H\,\mathbf{Z}\,\boldsymbol{\omega} \leq P \nonumber
\end{align}
where $\mathbf{Y}^T\,\otimes\,\mathbf{I}_{N_R}\,=\mathbf{Z}$. 
It is well know that this problem is generally non-convex and it can be approximated by semidefinite relaxation \cite{Palomar10}. 
\subsection{Approximation of the Non-Convex Fractional Program}
The fractional quadratic program \eqref{MBPDL1_twoway_simple} is non-convex and $\mathcal{NP}$-hard, in general \cite{Palomar10}. The state-of-the-art method to solve quadratically constrained fractional programs is 
a relaxation to a quasi-convex form based on a semi-definite program (SDP) \cite{Palomar10}. 
A bisection algorithm solves several convex feasibility check problems and converges arbitrarily closely to the global optimal value \cite{Boyd04}.
By dropping the non-convex rank-1 constraint, the feasibility check problem is given by a semi-definite program. A near optimal rank-1 solution can be recovered by a randomization method \cite{SIDIRO2006}.
Hence, the approximation of the optimal solution is based on semi-definite relaxation. This relaxation results in near optimal solutions, however, at the expense of high worst case complexity~\cite{KARIPIDIS08}.

In what follows a new approximation of the optimal solution is presented. The approximation is based on an estimation of the minimax upper bound of the optimal value. 
\newtheorem{lemma}{Lemma} 
\begin{lemma}\label{minimax_inequality} \cite{Boyd04}  Minimax inequality: Let $\mathcal{X}$ and $\mathcal{Y}$ be arbitrary sets and let $f()$ be an arbitrary function, then 
\begin{equation}
 \min_{\boldsymbol{y}\in \mathcal{Y}} \max_{\boldsymbol{x}\in \mathcal{X}}f(\boldsymbol{y},\boldsymbol{x}) \geq \max_{\boldsymbol{x}\in \mathcal{X}}\min_{\boldsymbol{y}\in \mathcal{Y}} f(\boldsymbol{y},\boldsymbol{x}).
\end{equation}

\end{lemma}
Let $\mathcal{J}=\{1,\ldots,N \}$, $N=2M$ be the index set of all SINRs, let $\mathbf{Q}_j$ and $\mathbf{P}_j$ be matrices indexed according to this new index set $\mathcal{J}$, and let $\mathcal{P}$ 
be the convex domain of $\boldsymbol{\omega}$ with $\mathcal{P}=\{\boldsymbol{\omega}\in \mathbb{C}^{N_R^2}\mid  \boldsymbol{\omega}^H\,\mathbf{Z}\,\boldsymbol{\omega} \leq P\}$, Problem \eqref{MBPDL1_twoway_simple} 
can be equivalently expressed by
\begin{align} \label{MBPDL1_twoway_simple2}
 \gamma^*= & \max_{\boldsymbol{\omega}\in \mathcal{P}}\;\;  \min_{j \in \mathcal{J}}  
\frac{\boldsymbol{\omega}^H\mathbf{Q}_{j}\boldsymbol{\omega}}{\boldsymbol{\omega}^H\mathbf{P}_j\boldsymbol{\omega}+\sigma^2}.
\end{align}
\begin{proposition} \label{Proposition_bound}
 Let $\lambda_{\text{max}}(\mathbf{A})$ be the largest eigenvalue of matrix $\mathbf{A}$ and let 
$\boldsymbol{F}_j=\mathbf{Z}^{-\frac{1}{2}}\mathbf{Q}_j\mathbf{Z}^{-\frac{1}{2}}$, $\bar{\mathbf{P}}_j=\mathbf{Z}^{-\frac{1}{2}}\mathbf{P}_j\mathbf{Z}^{-\frac{1}{2}}$, $\boldsymbol{G}_j=\bar{\mathbf{P}}_j+\frac{\sigma^2}{p}\mathbf{I}$\footnote{$p$ denotes the transmit power.},
the upper bound of \eqref{MBPDL1_twoway_simple2} is given by
\begin{align} \label{MBPDL1_twoway_bound}
 \bar{\gamma}= \min_{\substack{j \in \mathcal{J}}} \;\; &  \lambda_{\text{max}} (\boldsymbol{G}_j^{-1} \boldsymbol{F}_j )\geq \max_{\boldsymbol{\omega}\in \mathcal{P}}\;\;  \min_{j \in \mathcal{J}}  
\frac{\boldsymbol{\omega}^H\mathbf{Q}_{j}\boldsymbol{\omega}}{\boldsymbol{\omega}^H\mathbf{P}_j\boldsymbol{\omega}+\sigma^2}.
\end{align}
\end{proposition}
\begin{proof}
 The proof follows directly from Lemma \ref{minimax_inequality}:
\begin{align} \label{MBPDL1_twoway_bound2}
 \bar{\gamma}= \min_{j \in \mathcal{J}} \;\; \max_{\boldsymbol{\omega}\in \mathcal{P}}  
\frac{\boldsymbol{\omega}^H\mathbf{Q}_{j}\boldsymbol{\omega}}{\boldsymbol{\omega}^H\mathbf{P}_j\boldsymbol{\omega}+\sigma^2}\geq \max_{\boldsymbol{\omega}\in \mathcal{P}}\;\;  \min_{j \in \mathcal{J}}  
\frac{\boldsymbol{\omega}^H\mathbf{Q}_{j}\boldsymbol{\omega}}{\boldsymbol{\omega}^H\mathbf{P}_j\boldsymbol{\omega}+\sigma^2}.
\end{align}
Similar to the work of Havary-Nassab and et al. \cite{HavaryTSPROC08}, by introducing a new variable with unit norm, i.e. $\mathbf{w};~\sqrt{p}\,\mathbf{w}=\mathbf{Z}^\frac{1}{2}\boldsymbol{\omega}$ and the domain 
 $\mathcal{W}=\{ \mathbf{w}\in\mathbb{C}^{N_R^2}\mid ||\mathbf{w}||^2=1 \}$, \eqref{MBPDL1_twoway_bound2} can be recast into:
\begin{align} \label{gamma_bound_recast}
 \bar{\gamma}= \min_{\substack{j \in \mathcal{J}}} \;\; &  \max_{\mathbf{w} \in \mathcal{W}}
\frac{\mathbf{w}^H\boldsymbol{F}_j\mathbf{w}}{\mathbf{w}^H\boldsymbol{G}_j\mathbf{w}}.
\end{align}
It is argued in \cite{HavaryTSPROC08} that the objective function in \eqref{gamma_bound_recast} is non-decreasing w.r.t.  to $p$, thus, the maximum over $\mathbf{w}$ is attained at $p=P$.
The Matrices $\boldsymbol{G}_j$ are positive definite, therefore, the upper bound is expressed by special eigenvalue problem \eqref{MBPDL1_twoway_bound}.
\end{proof}
The upper bound of problem \eqref{MBPDL1_twoway_simple2} is a close bound. Regard Lemma \ref{minimax_inequality}, as argued in \cite{BarrosJGO1996}, in the case $\mathcal{Y}$ is a compact and convex set, 
$\mathcal{X}$ is a convex set and $f()$ is a real valued function, where $f(\boldsymbol{y},\cdot)$ is upper semi-continuous and quasi-concave on $\mathcal{X}$ 
for all $\boldsymbol{y}\in \mathcal{Y}$ and $f(\cdot,\boldsymbol{x})$ is lower semi-continuous 
and quasi-convex on $\mathcal{Y}$ for all $\boldsymbol{x}\in \mathcal{X}$, strong duality holds. 
Strong duality is not given for problem \eqref{MBPDL1_twoway_simple2} due to the non-convexity of the SINR function on $\mathcal{P}$ for all $\boldsymbol{y}\in\mathcal{Y}$.
Proposition \ref{Proposition_bound} provides a bound in the vicinity of the optimal value $\gamma^*$. 
The upper bound $\gamma^*$ is not reachable in general, however, it is possible to find an $\boldsymbol{\omega}^*$ which yields an SINR close to the optimal value $\gamma^*=\bar{\gamma}-\epsilon$ for some $\epsilon\geq 0$.

\section{Algorithm}
Proposition \ref{Proposition_bound} offers a direct solution for a close upper bound of the balanced SINR.
Compared to the work of Tao et al. \cite{TaoTSPROC2012}, this upper bound leads to an algorithm where the number of bisection iterations can be reduced.
Assuming the upper bound $\gamma$ is tight, $\gamma\approx \gamma^*$, or  $\gamma^*=\gamma-\epsilon$, the problem \eqref{MBPDL1_twoway_simple2} can be approximated by:
\begin{align} \label{MBPDL1_approx}
\text{find} & \;\;\mathbf{w}\in \mathcal{W} \\
\text{s.t.:} & \;\;\mathbf{w}^H(\boldsymbol{F}_j-\gamma\boldsymbol{G}_j)\mathbf{w}=0 \;\; \forall j\in \mathcal{J}.\nonumber
\end{align}
Let $\mathbf{D}_j(\gamma)=\boldsymbol{F}_j-\gamma\boldsymbol{G}_j$, Problem \eqref{MBPDL1_approx} is a nonlinear system of equations with 
$f_i(\mathbf{w})=\mathbf{w}^H\mathbf{D}_j(\gamma)\mathbf{w}$ and $\mathbf{f}(\mathbf{w})=[f_1(\mathbf{w}),\ldots,f_{2M}(\mathbf{w})]^T$.
Hence, we are interested in finding $\mathbf{f}(\mathbf{w})=\mathbf{0}$.
Using the complex real isomorphisms for a complex vector $\mathbf{y}\in\mathbb{C}^{n}$ and a real vector $\mathbf{x}\in\mathbb{R}^{2n}$
\begin{equation}
 \mathbf{x}=\hat{\mathbf{y}}=[\mathcal{R}\{\mathbf{y}\}^T,\mathcal{I}\{\mathbf{y}^T\}]^T
\end{equation}
and for a complex Matrix $\mathbf{Y}\in\mathbb{C}^{n\times n}$ and a real matrix $\mathbf{X}\in\mathbb{R}^{2n\times 2n}$
\begin{equation}
 \mathbf{X}=\hat{\mathbf{Y}}=
\begin{bmatrix}
\mathcal{R}\{\mathbf{Y}\} &-\mathcal{I}\{\mathbf{Y}\} \\
\mathcal{I}\{\mathbf{Y}\} & \mathcal{R}\{\mathbf{Y}\}                                                                                   
\end{bmatrix},
\end{equation}
Now, with $\hat{\mathbf{w}}\in \hat{\mathcal{W}}=\{ \hat{\mathbf{w}}\in\mathbb{R}^{2N_R^2}\mid ||\hat{\mathbf{w}}||^2=1 \}$and $\hat{\mathbf{D}}_j(\gamma)$ 
the function $f_i(\mathbf{w})=\mathbf{w}^H\mathbf{D}_j(\gamma)\mathbf{w}=\hat{\mathbf{w}}^T\hat{\mathbf{D}}_j(\gamma)\hat{\mathbf{w}}=f_i(\hat{\mathbf{w}})$ we have
$\mathbf{f}(\hat{\mathbf{w}})=\mathbf{0}$.
Multiple low complexity algorithms to find near optimal solutions of  to solve $\mathbf{f}(\hat{\mathbf{w}})=\mathbf{0}$ exist.
Several approaches are based on the Newton's method \cite{Kelley99}. Let  $\mathbf{B}_i(\gamma)=\hat{\mathbf{D}}_i(\gamma)+\hat{\mathbf{D}}_i^T(\gamma)$, the Jacobian matrix of $\mathbf{f}(\hat{\mathbf{w}})$ is
\begin{equation} \mathbf{\nabla}\mathbf{f}(\hat{\mathbf{w}})=
\begin{bmatrix}
( \mathbf{B}_1(\gamma)\hat{\mathbf{w}})^T\\
\hdots \\
(\mathbf{B}_{2M}(\gamma) \hat{\mathbf{w}})^T
\end{bmatrix}.
 \end{equation}
The Newton-like methods converge to a local optimal solution if the Lipschitz condition holds~\cite{Kelley99}.
\begin{lemma} \label{lipschitz}
 Let $K=\sqrt{\sum_{i=1}^{2M}\sum_{j=1}^{2N_R^2}|[\mathbf{B}^T_i(\gamma)]_{:,j}|^2}$, the function $\mathbf{f}(\hat{\mathbf{w}})$ is Lipschitz continuously differentiable.
\end{lemma}
\begin{proof}
 The Lipschitz condition for the Jacobian matrix $\mathbf{\nabla}\mathbf{f}(\hat{\mathbf{w}})$ is
\begin{equation}\label{Lipschitz}
 || \mathbf{\nabla}\mathbf{f}(\hat{\mathbf{w}}_1)-\mathbf{\nabla}\mathbf{f}(\hat{\mathbf{w}}_2)||\leq K||\hat{\mathbf{w}}_1-\hat{\mathbf{w}}_2||.
\end{equation}
The left side of \eqref{Lipschitz} can be rephrased as
\begin{equation}\label{Lipschitz2}
 || \mathbf{\nabla}\mathbf{f}(\hat{\mathbf{w}}_1)-\mathbf{\nabla}\mathbf{f}(\hat{\mathbf{w}}_2)||^2=\sum_{i=1}^{2M}\sum_{j=1}^{2N_R^2} |(\hat{\mathbf{w}}_1-\hat{\mathbf{w}}_2)^T[\mathbf{B}^T_i(\gamma)]_{:,j}|^2.
\end{equation}
Using the Cauchy-Schwarz inequality $|\mathbf{x}^T \mathbf{y}|^2\leq \mathbf{x}^T \mathbf{x} \cdot \mathbf{y}^T \mathbf{y}$, Eq. \eqref{Lipschitz2} is upper bounded by
\begin{align}\label{Lipschitz3}
 \sum_{i=1}^{2M}\sum_{j=1}^{2N_R^2} (\hat{\mathbf{w}}_1-\hat{\mathbf{w}}_2)^T(\hat{\mathbf{w}}_1-\hat{\mathbf{w}}_2) \cdot [\mathbf{B}^T_i(\gamma)]_{:,j}^T [\mathbf{B}^T_i(\gamma)]_{:,j}\\ 
=K^2 ||\hat{\mathbf{w}}_1-\hat{\mathbf{w}}_2||^2. \nonumber
\end{align}
Hence, the Jacobian $\mathbf{\nabla}\mathbf{f}(\hat{\mathbf{w}})$ is Lipschitz continuous. Consequently, $\mathbf{f}(\hat{\mathbf{w}})$ is Lipschitz continuously differentiable. 
\end{proof}

The Levenberg-Marquardt Method (LM) algorithm is an improved Newton based method to solve $\mathbf{f}(\hat{\mathbf{w}})=\mathbf{0}$ in a least squares sense. 
It prevents the Newton step to become unidentified because of a singular Jacobian matrix. 
The LM update is given by:
\begin{equation} \label{deltak}
 \boldsymbol{\delta}_k=-(\mathbf{\nabla}\mathbf{f}(\hat{\mathbf{w}}_k)^T\mathbf{\nabla}\mathbf{f}(\hat{\mathbf{w}}_k)+\mu_k \mathbf{I})^{-1}\mathbf{\nabla}\mathbf{f}(\hat{\mathbf{w}}_k)^T\mathbf{f}(\hat{\mathbf{w}}_k).
\end{equation}
Yamashita et al. \cite{Yamashita2001} have proved that $\mu_k=||\mathbf{f}(\hat{\mathbf{w}}_k)||^2$ provides super-linear convergence. 
Recently, Fan et al. \cite{FanSpringer2005} have extended the work of \cite{Yamashita2001} and proved that the parameter $\mu_k=||\mathbf{f}(\hat{\mathbf{w}}_k)||$ 
can achieve super-linear convergence if $||\mathbf{f}(\hat{\mathbf{w}}_k)||^{\delta}$, with $\delta \in [1,2]$ provides a local error bound.
\newtheorem{definition}{Definition}
\begin{definition}
Let $\hat{\mathbf{w}}\in \hat{\mathcal{W}}$ and let $\hat{\mathbf{w}}^*\in \hat{\mathcal{W}}^*$ be an optimal solution where $\hat{\mathcal{W}}^*$ is the set of optimal solutions. 
Let $\hat{\mathcal{W}}\cap \hat{\mathcal{W}}^*\neq \emptyset$, then $||\mathbf{f}(\hat{\mathbf{w}})||$ provides 
a local error bound on $\hat{\mathcal{W}}$ for \mbox{$\mathbf{f}(\hat{\mathbf{w}})=\mathbf{0}$} if there exists constant $c>0$ such that
\begin{equation}
 ||\mathbf{f}(\hat{\mathbf{w}})||\geq c \cdot \operatorname{dist}( \hat{\mathbf{w}},\hat{\mathbf{w}}^*)\;\; \forall \hat{\mathbf{w}}\in \hat{\mathcal{W}},\;\forall \hat{\mathbf{w}}^*\in\hat{\mathcal{W}}^*.
\end{equation}
\end{definition}
  \begin{figure}[t]
  \begin{center}
      \includegraphics[scale=0.8]{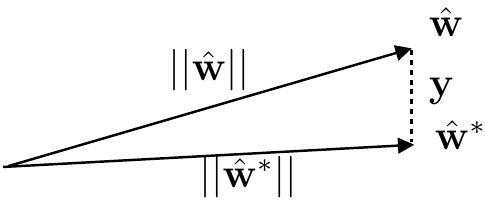}
      \caption{\small Distance between the solution $\hat{\mathbf{w}}$ and the optimal solution $\hat{\mathbf{w}}^*$.}
      \label{w_fig}
  \end{center}
 \end{figure}
A solution $\mathbf{f}(\hat{\mathbf{w}}^*)=\mathbf{0}$ can be achieved, e.g., by power control at the users.
In case a nonempty solution unequal $\hat{\mathbf{w}}=\mathbf{0}$ exists, we can proof that the algorithm converges superlinearly to the optimal balanced SINR when we are very close to the optimal solution.
It is hard to prove that the local error bound exists for every given value of $\gamma$. However, we can prove it for a $\gamma$ such that we have for at least one user $j^*$ that $\hat{\mathbf{D}}_{j^*}(\gamma)\prec 0$.
\begin{proposition} \label{localerrorbound} 
We assume having a tight upper bound of the SINR $\gamma=\bar{\gamma}$ such that we have least one user $j^*$ where \mbox{$\hat{\mathbf{D}}_{j^*}(\gamma)=\hat{\boldsymbol{F}}_{j^*}-\gamma\hat{\boldsymbol{G}}_{j^*}\prec 0$} and
the initial solution $\hat{\mathbf{w}}_0$, with $||\hat{\mathbf{w}}_0||=1$ of the Levenberg-Marquardt algorithm is sufficiently close to $\mathcal{W}^*$, $\operatorname{dist}( \hat{\mathbf{w}},\hat{\mathbf{w}}^*)< b$, with $b<1$ 
and $\mathbf{f}(\hat{\mathbf{w}}^*)=\mathbf{0}$ has a nonempty solution set.
Furthermore, let $\mu_k=||\mathbf{f}(\hat{\mathbf{w}}_k)||$, then sequence $\{\hat{\mathbf{w}}_{k+1}=\hat{\mathbf{w}}_k+\boldsymbol{\delta}_k\}$ converges superlinearly.
\end{proposition}
\begin{proof}
First, we have to prove that $||\mathbf{f}(\hat{\mathbf{w}})||$ provides a local error bound. 
As shown in Fig.~\ref{w_fig}, we have 
 \begin{align} 
 &||\hat{\mathbf{w}}||=1 > b > ||\mathbf{y}|| =||\hat{\mathbf{w}}-\hat{\mathbf{w}}^*||. \label{w_inequality}
 \end{align}
The function $||\mathbf{f}(\hat{\mathbf{w}})||$, is lower bounded by:
\begin{align*} 
||\mathbf{f}(\hat{\mathbf{w}})||&=\sqrt{\sum_{\substack{j=1}}^{2M} |\hat{\mathbf{w}}^T\hat{\mathbf{D}}_j(\gamma)\hat{\mathbf{w}}|^2 }\geq \sqrt{ |\hat{\mathbf{w}}^T\hat{\mathbf{D}}_{j^*}(\gamma)\hat{\mathbf{w}}|^2} \geq 0 \nonumber
\end{align*}
where $j^*$ denotes the mentioned selected user index. We always have:
\begin{align*}
\lambda_{\text{min}}(\hat{\mathbf{D}}_{j^*}(\gamma))||\hat{\mathbf{w}}||^2 \leq \hat{\mathbf{w}}^T\hat{\mathbf{D}}_{j^*}(\gamma)\hat{\mathbf{w}} \leq \lambda_{\text{max}}(\hat{\mathbf{D}}_{j^*}(\gamma))||\hat{\mathbf{w}}||^2.
\end{align*}
Due to $\hat{\mathbf{D}}_{j^*}(\gamma)\prec 0 $, we have
\begin{align*}
||\hat{\mathbf{w}}||^2|\lambda_{\text{max}}(\hat{\mathbf{D}}_{j^*}(\gamma))| \leq |\hat{\mathbf{w}}^T\hat{\mathbf{D}}_{j^*}(\gamma)\hat{\mathbf{w}}| \leq||\hat{\mathbf{w}}||^2 |\lambda_{\text{min}}(\hat{\mathbf{D}}_{j^*}(\gamma))|.
\end{align*}
and we can use the inequality
\begin{equation*}
 ||\mathbf{f}(\hat{\mathbf{w}})||\geq |\hat{\mathbf{w}}^T \hat{\mathbf{D}}_{j^*}(\gamma)\hat{\mathbf{w}}| \geq ||\hat{\mathbf{w}}||^2\,|\lambda_{\text{max}}(\hat{\mathbf{D}}_{j^*}(\gamma))|=||\hat{\mathbf{w}}||^2\, c. 
\end{equation*}
Using the inequality:
\begin{equation*}
 1=||\hat{\mathbf{w}}||=||\hat{\mathbf{w}}||^2 >||\mathbf{y}||,
\end{equation*}
we have: $||\mathbf{f}(\hat{\mathbf{w}})||\geq ||\mathbf{y}||\,c.$
Using \eqref{w_inequality} we have
\begin{align*} 
&||\mathbf{f}(\hat{\mathbf{w}})|| \geq c\cdot ||\hat{\mathbf{w}}-\hat{\mathbf{w}}^*||.
\end{align*}
 According to Lemma \ref{lipschitz}, $||\mathbf{f}(\hat{\mathbf{w}})||$ is Lipschitz continuously differentiable, consequently the two assumptions of \cite[Theorem 2.1]{FanSpringer2005} hold
and  $\{\hat{\mathbf{w}}_{k+1}=\hat{\mathbf{w}}_k+\boldsymbol{\delta}_k\}$ converges superlinearly.
\end{proof}

Having a well chosen $\gamma$ and its corresponding eigenvector $\hat{\mathbf{w}}=\hat{\mathbf{w}}_0$ as initial solution and assuming $\hat{\mathbf{w}}_0$ 
is close to the set of solutions $\hat{\mathcal{W}}^*$ satisfying \eqref{MBPDL1_approx} leads to a fast convergence of the classical LM algorithm with $\mu_k=||\mathbf{f}(\hat{\mathbf{w}}_k)||$. 
Several simulation runs have shown a fast convergence if $\hat{\mathbf{w}}_0$ is chosen based on the upper bound \eqref{gamma_bound_recast} with a sufficiently large $\epsilon$.
However, in some cases, the LM method still requires a lot of iterations. A fast convergence to a local optimal solution is not guaranteed.

Therefore, this paper uses a modified LM algorithm based on a line search to find the optimal step size. Firstly, the unconstrained case ($\hat{\mathbf{w}}\in\mathbb{R}^{2N_R^2}$) is considered. 
Algorithm~\ref{MLM} presents the outline of the used modified LM method.
\begin{algorithm}
\caption{Modified Levenberg-Marquardt Method}
\label{MLM}
\begin{algorithmic}
\STATE \textbf{Initialize:} Find a $\hat{\mathbf{w}}_{0}$ based on the upper bound \eqref{gamma_bound_recast}, set $\epsilon>0$ sufficiently large. 
Set a $\nu \in (0,1)$ and the accuracy $\epsilon_{LM}$. Set $k=0$.
\WHILE{$||\mathbf{\nabla}\mathbf{f}(\hat{\mathbf{w}}_k)^H\mathbf{f}(\hat{\mathbf{w}}_k)||\geq \epsilon_{LM}$ and $k<N_{max}$}
\STATE Set $\mu_k=||\mathbf{f}(\hat{\mathbf{w}}_k)||$ and compute $\boldsymbol{\delta}_k$ by \eqref{deltak}
\IF{$||\mathbf{f}(\hat{\mathbf{w}}_k+\boldsymbol{\delta}_k)||\leq \nu ||\mathbf{f}(\hat{\mathbf{w}}_k)||$ }
\STATE $\hat{\mathbf{w}}_{k+1}=\hat{\mathbf{w}}_k+\boldsymbol{\delta}_k$
\ENDIF
\STATE Compute step size $\alpha_k$ by Armijo line search \cite{FanSpringer2005}.
\STATE $\hat{\mathbf{w}}_{k+1}=\hat{\mathbf{w}}_k+\alpha_k\boldsymbol{\delta}_k$ and $k=k+1$
\ENDWHILE
\STATE $\mathbf{w}_{k+1}\;\leftarrow \; \text{construct complex vector from } \hat{\mathbf{w}}_{k+1}$
\STATE $\mathbf{w}_{k+1}\;\leftarrow \;\sqrt{P}\mathbf{Z}^{-1/2}\mathbf{w}_{k+1}$
\RETURN $\mathbf{w}_{k+1}$
\end{algorithmic}
\end{algorithm}
\begin{proposition}
\cite[Theorem 3.1]{FanSpringer2005} Let the sequence $\{\hat{\mathbf{w}}_k\}$ be generated by Alg. \ref{MLM} with line search. 
Then any accumulation point of the sequence $\{\hat{\mathbf{w}}_k\}$ is a stationary point of $1/2||f(\hat{\mathbf{w}})||^2$. 
If an accumulation point of the sequence $\{\hat{\mathbf{w}}_k\}$ is a solution of \eqref{MBPDL1_approx}, then $\{\hat{\mathbf{w}}_k\}$ converges to the solution quadratically.
\end{proposition}
\begin{proof}
 The proof is straightforward. As in the proof of Proposition \ref{localerrorbound}, the assumptions of \cite[Theorem 3.1]{FanSpringer2005} are already satisfied, 
if $\hat{\mathbf{w}}_0$ is sufficiently close to a solution satisfying \eqref{MBPDL1_approx}. 
In this case the algorithm converges  according to \cite[Theorem 3.1]{FanSpringer2005}.
\end{proof}
Algorithm~\ref{MLM} has the advantage of a fast convergence if the initial solution is close to the set of solutions satisfying \eqref{MBPDL1_approx}. 
In the other cases, Alg.~\ref{MLM} still converges to a least squares solution \cite{FanSpringer2005,Yamashita2001}.

\section{Numerical Results}
\label{simulation_comment}We optimized the precoding vectors with the presented algorithms and calculated the achievable rate \cite{TaoTSPROC2012,BournakaGC2011,SchadCAMSAP2011}.
To justify the efficiency of the proposed methods a huge number of simulations ($1000$) each with a different realization of channel coefficients are generated. 
In these simulations the number of users is chosen 
to be $2M=6$ while an RS  with $N_R=6$ antennas is assumed. The SNR in MAC phase is chosen to be constant, $SNR_{MAC}=10\log(\frac{P}{\sigma^2_R})=10$dB. 
Then we have varied peak power to noise ratio, $10\log(\frac{P}{\sigma^2})$ as in \cite{TaoTSPROC2012} and $P=10$. Also, the channel coefficients are assumed to be Rayleigh distributed. The channels are generated 
similar to \cite{Narasimhan2006}. First we have generated channel matrices, $\mathbf{H}_t=\tilde{\mathbf{H}}_t\mathbf{T}_t$, with entries in $\tilde{\mathbf{H}}_t$ which are i.i.d Gaussian random variables with zero means and 
unit variances and $\mathbf{T}_t=\sqrt{\frac{P}{M}}\mathbf{I}$. Then,
 we have made them correlated in order to preserve more practical relevance as follows:
$\tilde{\mathbf{H}}_1=\boldsymbol{\Theta}_{RS}^{1/2}\,\tilde{\mathbf{H}}_1\,\boldsymbol{\Theta}_1^{1/2}, \; \; \;\; \tilde{\mathbf{H}}_2=\boldsymbol{\Theta}_{RS}^{1/2}\,\tilde{\mathbf{H}}_2\,\boldsymbol{\Theta}_2^{1/2}$
where $[\boldsymbol{\Theta}_2]_{ij}=(\rho_{2})^{|i-j|}$, $[\boldsymbol{\Theta}_{RS}]_{ij}=(\rho_{RS})^{|i-j|}$ and $[\boldsymbol{\Theta}_1]_{ij}=(\rho_{1})^{|i-j|}$. The value
 $\rho_{RS}=0.5$ is chosen for RS antennas while users  are assumed to be less correlated than RS since they are spatially distributed within the cell, i.e. $\rho_{1}=\rho_{2}=0.1$. The numerical 
results are generated for the following methods:
\begin{itemize}
 \item Semidefinite relaxation based bisection algorithm as in \cite{TaoTSPROC2012} with  $\epsilon_{BS}=10^{-7}$ to calculate a tight a upper bound.
 \item LM method based bisection: Here the search for intial solutions for $\gamma$ is based on a bisection, to get into the vincinity ($\epsilon=0.1$) of the optimal SINR $\gamma^*$. 
If the minimum SINR of all users is larger in the next bisection step, the search will continue in the upper half interval else it continuous in the lower interval.
The LM method has the following parameter configuration: $\nu=0.9$, $\epsilon_{LM}=10^{-7}$, $N_{max}=50$, $\alpha_0=0.25$, and the upper bound of $\gamma$ is  scaled with an $\delta \in [0.6,\ldots,1]$ depending in the SNR 
to speed up the bisection search.
\end{itemize}
  \begin{figure}[t]
  \begin{center}
      \includegraphics[scale=0.35]{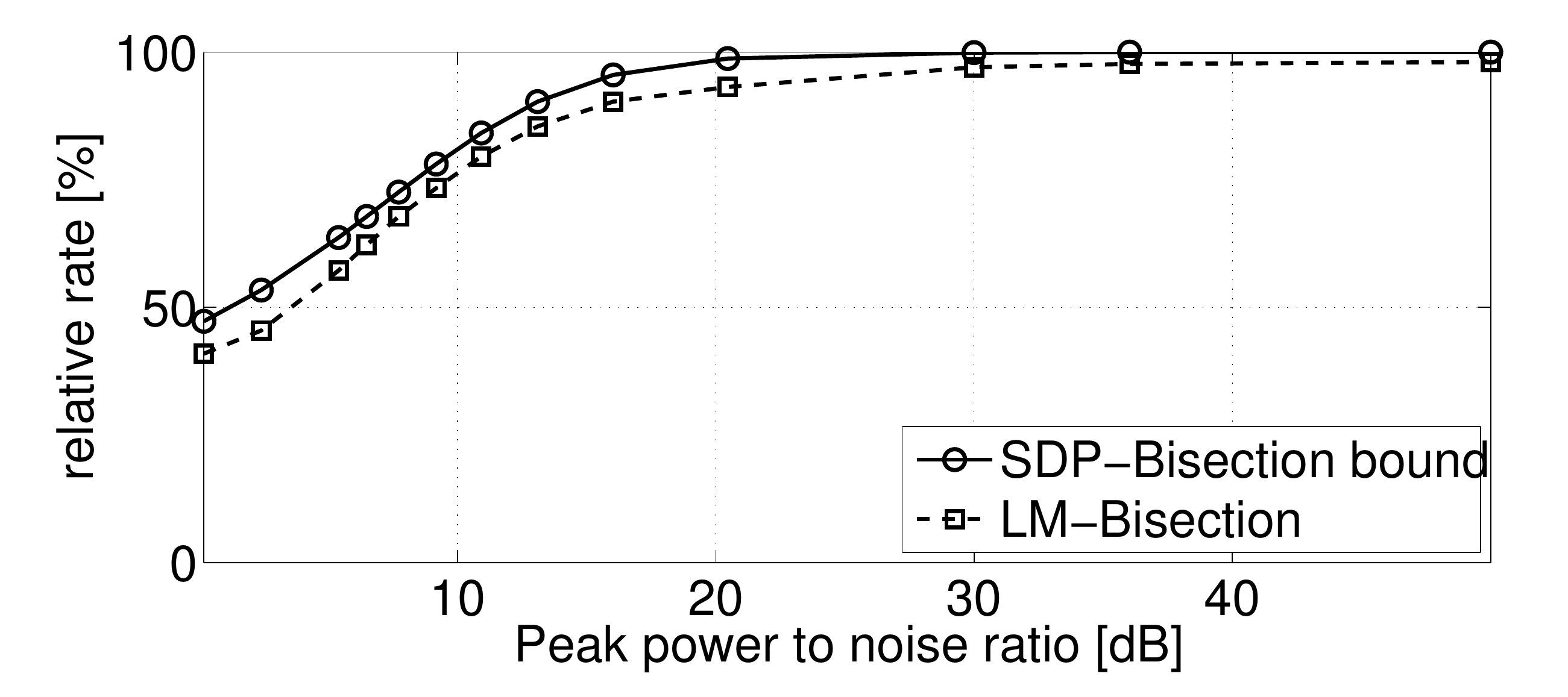}
      \caption{\small Numerical results for the minimum (achievable) rate of the investigated algorithms relative to the minimax upperbound in percent for different peak power to noise ratios.}
      \label{rate_fig}
  \end{center}
 \end{figure}

  \begin{figure}[t]
  \begin{center}
      \includegraphics[scale=0.35]{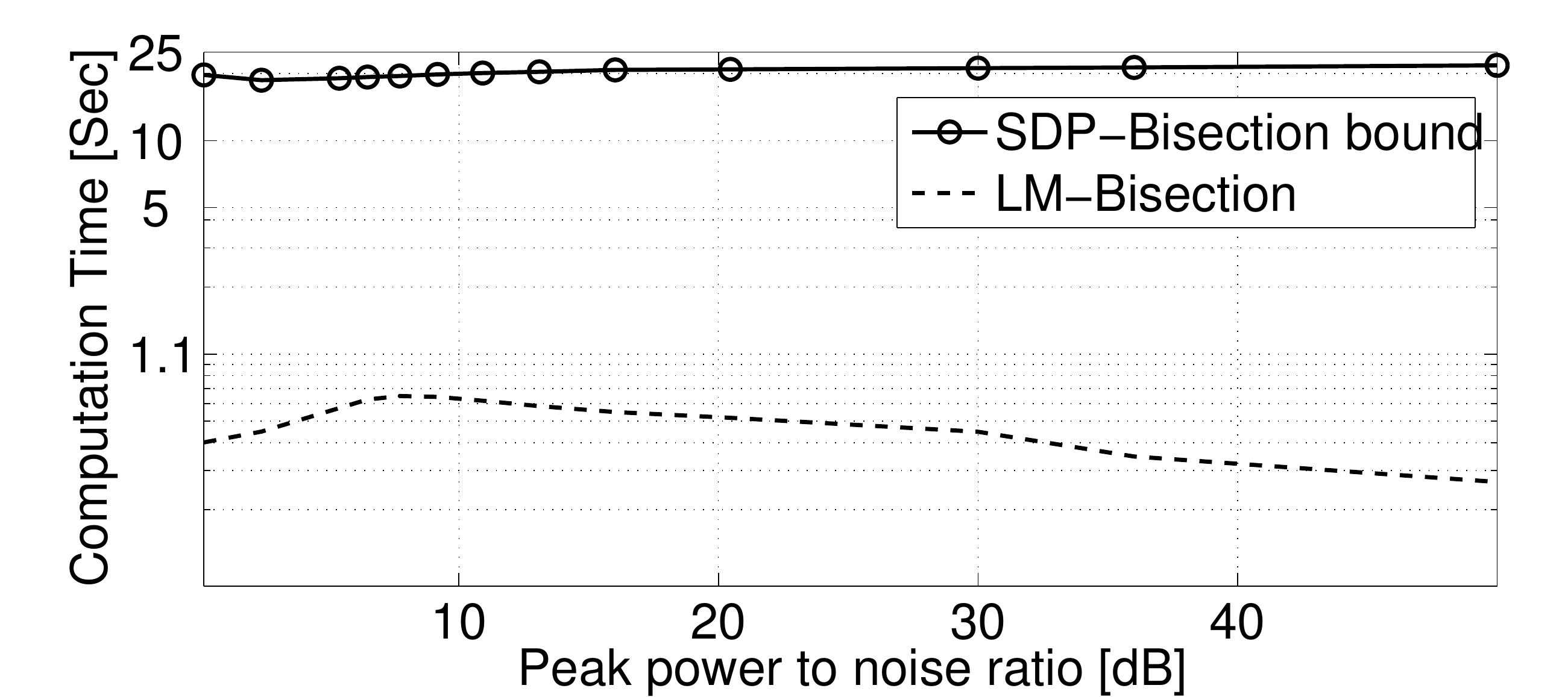}
      \caption{\small Numerical results for the computation time for different peak power to noise ratios.}
      \label{comp_fig}
  \end{center}
 \end{figure}

  \begin{figure}[t]
  \begin{center}
      \includegraphics[scale=0.35]{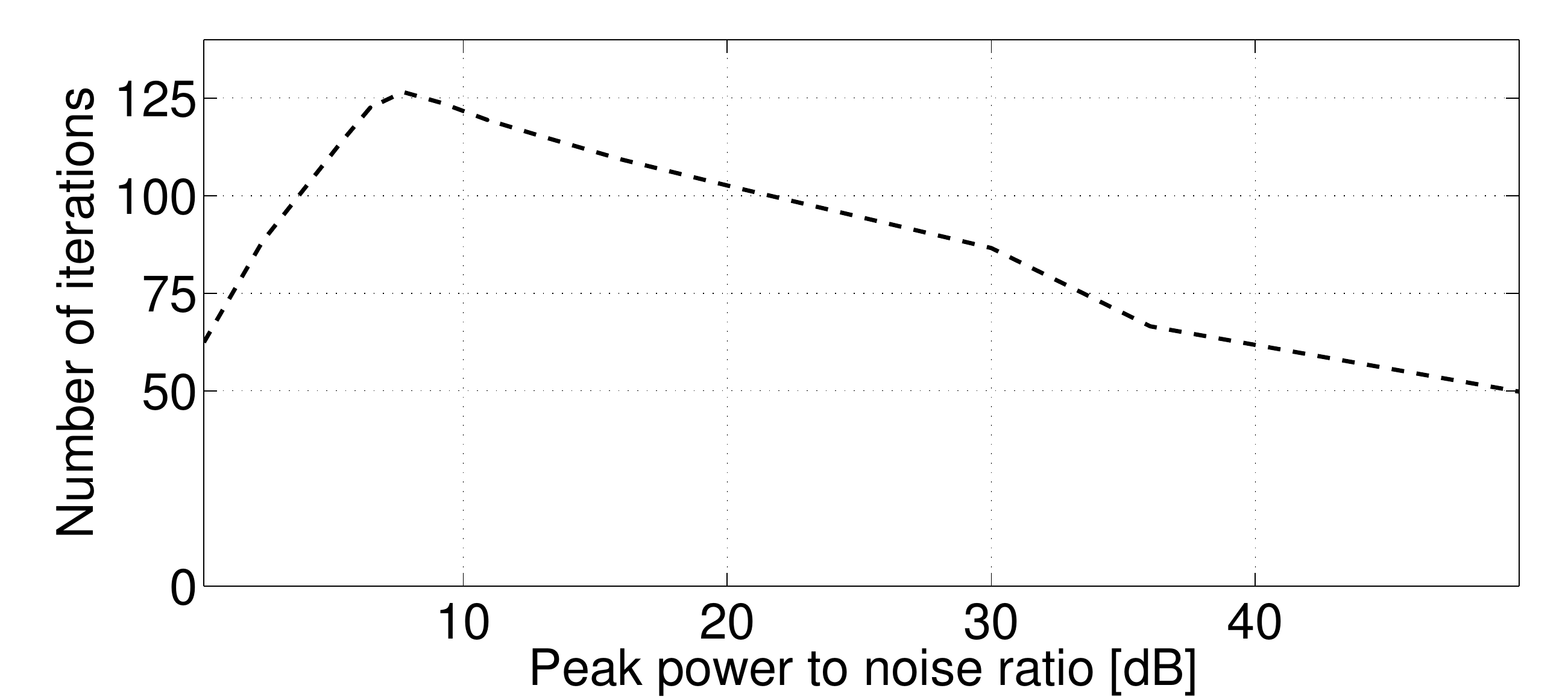}
      \caption{\small Numerical results for the mean total number of iterations for different peak power to noise ratios. Here the total number of iterations is shown. 
Also for the LM bisection algorithm all LM iterations in each bisection step are summed up. }
      \label{it_fig}
  \end{center}
 \end{figure}

  \begin{figure}[t]
  \begin{center}
      \includegraphics[scale=0.35]{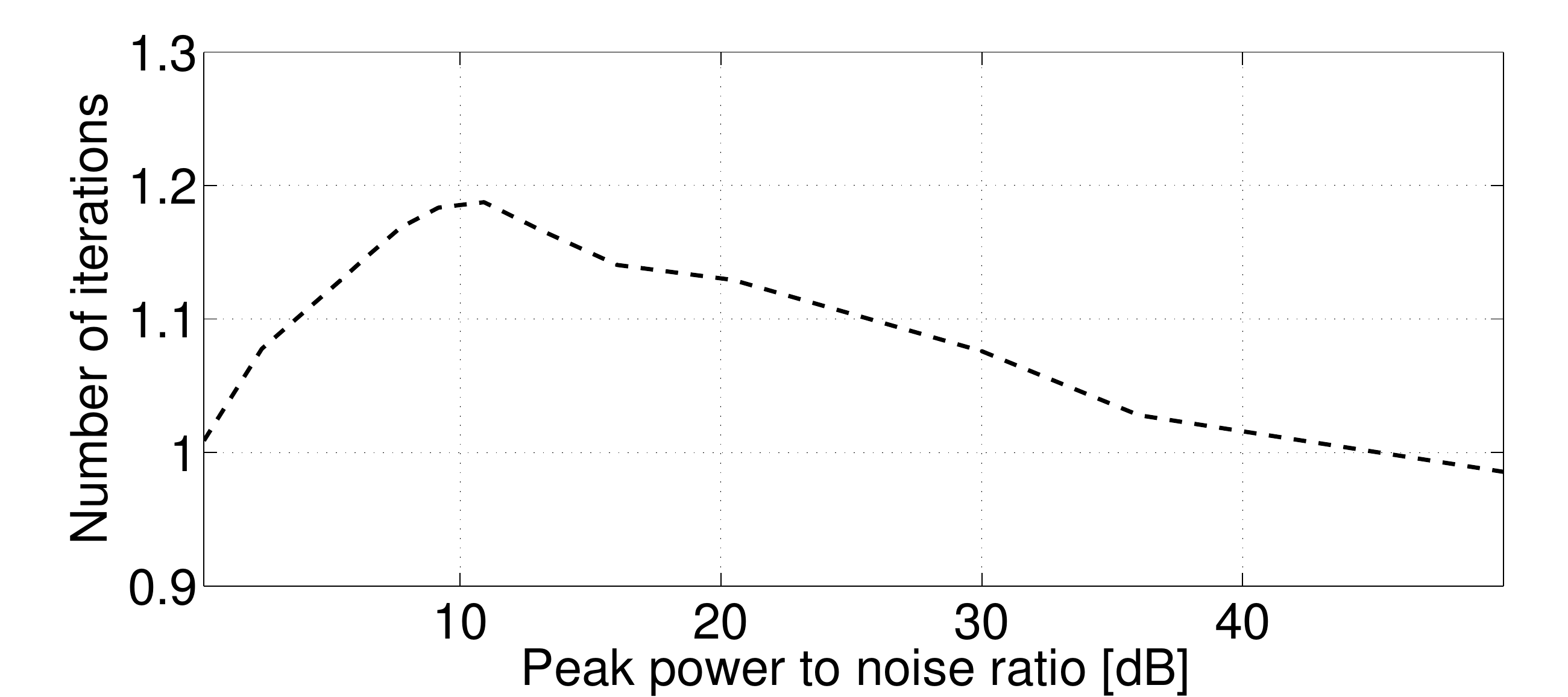}
      \caption{\small Numerical results for the mean number of line search iterations per LM step for different peak power to noise ratios.}
      \label{it_ls_fig}
  \end{center}
 \end{figure}

Figure \ref{rate_fig} shows the minimum (achievable) user rate  ($1/2 \log_{2}(1+ \gamma_i^t)$) of the different algorithms relative to the minimax upperbound in percent. 
As it can be observed, the LM method achieves rates close to the upper bound based on the SDP. 
The upper bound of Proposition \ref{Proposition_bound} is very tight in high SNR.
Figure \ref{it_fig} depicts the mean total number of iterations for the LM methods for different SNR values. Especially in low and very high SNR, the LM method converges fast.
Figure \ref{it_ls_fig} shows a fast convergence of the line search adaptation as well.

Let $n$ be the variable size, a sedumi-based \cite{Sedumi09} SDP solution for a multicast beamforming scenario with a worst case complexity of $\mathcal{O}(n^6)$ per iteration was proposed in \cite{SIDIRO2006}. 
A quadratic programming based SDP can be faster solved in $\mathcal{O}(n^{4.5})$ \cite{LuoMagazin10}.
The LM approach has complexity $\mathcal{O}(n^3)$ due to the matrix inversion. The computational complexity can be further reduced if a direct implementation in the complex domain is used. Such an implementation 
shows similar results.
Figure \ref{comp_fig}, shows the computation time of the two presented methods. It has to be emphasized that the SDP-based approach uses optimized code and the proposed LM-based approach uses not optimized code. 
However, the new proposed LM-method uses much less computation time than the conventional SDP-based technique.

\section{Conclusion}
This paper presents a novel approach for a low complexity algorithm for the non-convex max-min SINR optimization problem in the bidirectional relay channel. 
The algorithm is based on a novel closed form solution of the upper bound and a modified Levenberg-Marquardt algorithm. The convergence of the new method is proved. 
Numerical results indicate the performance of the proposed method. 
The achievable rate of the new algorithm is very close to the upper bound.

A fast convergence can be achieved if the initial solution $\hat{\mathbf{w}}_{0}$ is very close to the optimal solution. A future work can be an improved search for an initial solution or an adaptation of individual SINR
constraints such that $\mathbf{f}(\hat{\mathbf{w}}^*)=\mathbf{0}$ can be always achieved.

\bibliographystyle{IEEEtran}

\bibliography{TVT_Correspondence_TWC.bib}

\end{document}